\documentclass[12pt]{article}
\usepackage[mathscr]{eucal}
\usepackage{amsmath,amsfonts,amssymb,amsthm}
\usepackage[matrix,arrow]{xy}
\usepackage{times}
\usepackage{float}
\usepackage{graphicx}
\usepackage{epstopdf}
\usepackage{multibox}
\usepackage{dcolumn}

\advance\textheight\topskip \textwidth 36pc \oddsidemargin 20pt
\numberwithin{equation}{section} \makeatletter
\@addtoreset{equation}{section}

\newtheorem{prop}{Proposition}[section]

\newtheorem{defn}{Definition}[section]

\begin{document}

\def\mytitle{Rigid Symmetries and Conservation Laws \\ in Non-Lagrangian Field Theory}

\pagestyle{myheadings}
\markboth{\textsc{\small D. Kaparulin, S. Lyakhovich and A. Sharapov}}{%
  \textsc{\small Symmetries and Conservation Laws }}
\addtolength{\headsep}{4pt}

\begin{flushright}\small
\end{flushright}
\vspace{2cm}
\begin{centering}

  \vspace{1cm}

  \textbf{\Large{\mytitle}}

  \vspace{1.5cm}

  {\large D.S. Kaparulin, S.L. Lyakhovich and A.A. Sharapov}

\vspace{.5cm}

\begin{minipage}{.9\textwidth}\small \it \begin{center}
   Department of Quantum Field Theory, Tomsk State University,\\
   Lenin ave. 36,  Tomsk 634050, Russia
   \end{center}
\end{minipage}

\end{centering}

\vspace{1cm}

\begin{center}
  \begin{minipage}{.9\textwidth}
\textsc{Abstract}. Making use of the Lagrange anchor construction
introduced earlier to quantize non-Lagrangian field theories, we
extend the Noether theorem beyond the class of variational
dynamics.
  \end{minipage}
\end{center}


\vfill

\noindent \mbox{}
\raisebox{-3\baselineskip}{%
  \parbox{\textwidth}{\mbox{}\hrulefill\\[-4pt]}}
{\scriptsize$^*$ This work was partially supported by the RFBR
grant 09-02-00723-a,  by the grant from Russian Federation
President Programme of Support for Leading Scientific Schools no
3400.2010.2., by the State Contract no 02.740.11.0238 from Russian
Federal Agency for Science and Innovation, and also by Russian
Federal Agency of Education under the State Contracts no P1337 and
no P22. SLL appreciates partial support from the RFBR grant
08-01-00737-a. }

\thispagestyle{empty}
\newpage

\begin{small}
{\addtolength{\parskip}{-1.5pt}
 \tableofcontents}
\end{small}
\newpage

\section{Introduction}
In this paper we extend the Noether theorem beyond the class of
variational dynamics.

The classical field theory is completely defined by equations of
motion. While  the least action principle is not vital at
classical level, it provides a number of useful tools for studying
classical dynamics. In particular, given an action functional
enjoying  infinitesimal symmetries, one can derive conserved
currents. For non-variational equations, the infinitesimal
symmetries do not necessarily result in  the conservation laws and
vice versa. All the symmetries are divided into gauge and rigid.
The gauge  symmetries are unambiguously connected to Noether's
identities (also called the strong conservation laws) as far as
equations are variational. If the equations do not come from the
least action principle, this relationship (sometimes called the
second Noether theorem) is generally invalid: It is possible to
have gauge invariant field equations without Noether's identities
or equations possessing strong conservation laws without being
gauge invariant.

In the gauge field theory, the existence of action is a key
prerequisite for constructing the standard BRST theory either in the
Hamiltonian BFV form \cite{BFV} or in the BV field-anti-field
formalism \cite{BV} (see \cite{HT} for review). The BRST theory, in
its turn, provides the most general tools for quantizing gauge
theories \cite{BFV, BV, HT}. Also, the BRST formalism serves for
the study of various classical problems such as constructing
consistent interactions in gauge models \cite{H1} or identifying the
nontrivial conservation laws and rigid symmetries \cite{BBH}.

Now, it is known \cite{LS1} that the BRST theory can be extended
beyond the class of dynamical systems whose equations of motion
are variational. Classically, the BRST complex can always be
constructed for any regular system of field equations, variational
or not \cite{HT}, \cite{LS1}, \cite{KLS}. The quantization,
however, requires an extra ingredient, besides the field
equations. In the framework of the deformation quantization it is
the \textit{weak Poisson structure} \cite{LS1} (see also
\cite{CaFe}). The existence of the weak Poisson structure is much
less restrictive for the evolutionary equations than the
requirement to be Hamiltonian. In the framework of the
path-integral quantization, the corresponding extra ingredient was
first introduced in \cite{KLS} and called the \textit{Lagrange
structure}. Again,  the existence of the Lagrange structure is
much less restrictive for the equations than the requirement to be
variational or admit an equivalent variational reformulation.
Whenever the Lagrange structure is known for a given system of
equations, the theory can be consistently quantized in three
equivalent ways. First, the original non-Lagrangian field theory
in $d$ dimensions can be converted  into equivalent Lagrangian
topological field theory in $d+1$. The latter can be then
quantized by the usual BV procedure \cite{KLS}. Second, the
generating functional of Green's functions can be defined through
the generalized Schwinger-Dyson equation \cite{LS2}, \cite{LSDUY}.
Third, the original non-Lagrangian field theory can be embedded
into an augmented Lagrangian theory in the same $d$ dimensions
\cite{LS3}.

Besides providing the basis for quantizing non-Lagrangian field
theories, the Lagrange structure binds together gauge symmetries
and the strong conservation laws, although the relationship is
more relaxed than for the variational case \cite{KLS}. In the
first-order formalism, for example,  the constraints and gauge
symmetries remain fully unrelated \cite{LSDB} unless the equations
are Hamiltonian. A certain  correlation is set up by the  weak
Poisson structure, although it is not so rigid as the relation
between the first-class constraints and gauge symmetries in the
Hamiltonian constrained dynamics.

In this paper we will show that the Lagrange structure allows one
to connect the rigid symmetries with the conservation laws. This
relationship reveals itself in a different way than the link
between the gauge symmetries and the strong conservation laws in
non-Lagrangian theories equipped with the Lagrange structure. In
the next sections we give accurate definitions for symmetries,
identities and conservation laws and, by making use of the
Lagrange structure,  establish a relationship between them. The
general construction is then exemplified by a class of
non-Lagrangian models, where some fundamental symmetries and/or
conserved currents can be explicitly found and compared.

To give a preliminary impression of the connection between the
symmetries and conservation laws, below in the Introduction, we
reformulate the standard Noether theorem in the way that
illuminates the points which are important for extending the
theorem beyond the class of Lagrangian dynamics.

Let $M$ denote the configuration space of fields $\varphi^i$.
Hereinafter we use De Witt's condensed notation \cite{DW} whereby
the index $i$, labelling the fields, includes also the space-time
coordinates $\{x^\mu \}$.  The smooth functions $f(\varphi)$ on
$M$ are identified with the local functionals of fields. In terms
of the condensed notation, the Lagrangian equations of motion are
given by the components of an exact 1-form $\mathrm{d}S$ on $M$,
where $\mathrm{d}=\delta\varphi^i\frac{\delta}{\delta\varphi^i}$
is the variational differential and $S$ is the action functional.
In general, the field equations can well be the components of a
section of another bundle over $M$, not necessarily $T^\ast M$.
For example, the tensor type of equations can differ from that of
fields (see the next section). Here, however, we restrict our
discussion to equations whose left hand sides are given by the
components of a local 1-form $T=T_i(\varphi)\delta \varphi^i$ on
$M$. Then the necessary condition for the existence of an action
for the field equations $T_i(\varphi)=0$ requires $T$ to be a
closed 1-form
\begin{equation}\label{dT}
\mathrm{d} T = 0 \,.
\end{equation}
This is the well-known Helmholtz criterion  from the inverse
problem of the calculus of variations.

Recall that a local vector field $\Psi =
\Psi^i(\varphi)\frac{\delta}{\delta\varphi^i} $ on $M$ is called a
characteristic of the equations $T_i=0$  if
\begin{equation}\label{dPsi}
\mathrm{d} ( i_\Psi T ) \equiv 0 \, , \qquad i_\Psi T \equiv
\Psi^i  T_i\,.
\end{equation}
Since the value $i_\Psi T$ is annihilated by the variational
derivative, it must be the integral of a total divergence
$\partial_\mu j^\mu$. By definition, the latter vanishes on the
equations of motion, and hence $j^\mu$ is a conserved current.

Every vector field $\Psi$ generates a transformation
$\delta\varphi^i=\Psi^i(\varphi)$  of the space of fields $M$ and
we say that $\Psi$ is a \textit{proper} symmetry of the equations
of motion if
\begin{equation}\label{Sym}
 L_\Psi T=0 \, ,
\end{equation}
where $L_\Psi$ is the Lie derivative along $\Psi$. For the
variational equations of motion $T=\mathrm{d}S$ the proper
symmetries are clearly the symmetries in the usual sense. Making
use of Cartan's formula
\begin{equation}\label{Cart}
  L_\Psi T= i_\Psi (\mathrm{d} T)  + \mathrm{d} ( i_\Psi T) \, ,
\end{equation}
we see that the proper symmetry $\Psi$ is a characteristic if
$i_\Psi (\mathrm{d} T) =0 $. For the variational equations of
motion, satisfying the Helmholtz condition (\ref{dT}), the last
formula identifies symmetries with characteristics. It is the
statement which is known as the Noether theorem.

Let us summarize the lessons that can be learned  from the above
formulation of the Noether theorem towards its extension beyond
the scope  of variational dynamics. The construction of
characteristics from symmetries involves, besides the equations of
motion, one more crucial ingredient, namely, the variational
exterior differential $\mathrm{d}$ that maps the variational
$p$-forms to the $(p+1)$-forms. In the particular case discussed
above, the equations of motion are given by 1-forms. Three main
properties of the differential $\mathrm{d}$ are relevant for
establishing the link between the symmetries and conservation
laws. First, the field equations are $\mathrm{d}$-closed
(\ref{dT}). Second, the exterior  differential and the Lie
derivative are related to each other through Cartan's formula
(\ref{Cart}). Third, the kernel of the variational differential is
constituted by the integrals of  total divergencies $\partial_\mu
j^\mu$. These three facts taken together imply that the symmetries
can be identified with the characteristics.

The general field equations need not be a section of the cotangent
bundle $T^\ast M$.  They may well be associated  with  another
vector bundle $\mathcal{E}\rightarrow M$ over the configuration
space of fields $M$, in which case the differential $\mathrm{d}$
does not work any more. Thus, one needs a more general operator
$\mathrm{d}_\mathcal{E}$ to replace the variational exterior
differential in the case of  non-Lagrangian equations of motion.
The aforementioned  Lagrange structure provides such an operator
that satisfies the first two properties. The third one is obeyed,
in a sense, only by weakly transitive Lagrange structures. Thus,
whenever the non-Lagrangian field equations enjoy a weakly
transitive Lagrange structure, the Noether theorem can be still
formulated in full. In the intransitive case, a weaker proposition
can be proved, connecting rigid symmetries with characteristics in
a more relaxed way.

\section{Classical fields}

As is well known, the geometry underlying  the cinematics of local
field theory is that of jet bundles. So we begin the section with
recalling some basics from the jet theory  relevant for our
purposes. For a systematic exposition  of the subject the reader
may consult the books \cite{KMS,Saund,Olv,KLV}.

The starting point of any local field theory is a locally trivial
fiber bundle $\pi: Y\rightarrow X$, which base $X$ is usually
identified with the space-time manifold and which sections,  i.e.,
the smooth maps $\varphi: X\rightarrow Y$ satisfying  $p\circ
\varphi=\mathrm{id}_X$, are called \textit{fields} or
\textit{histories}. The space of all histories (= sections of $Y$)
will be denoted by $M$. The typical fiber of $Y$, being given by
$U\simeq\pi^{-1}(x)$, $x\in X$, is usually referred to as the
\textit{target space of fields}.

The $r$-th jet prolongation of the fiber bundle $Y\rightarrow X$
is the bundle $\pi_r: J^rY\rightarrow X$, whose points are
$r$-jets. By definition, an $r$-jet  $j_x^r\varphi$ is the
equivalence class of a local section $\varphi$ of $Y$, where two
local sections $\varphi$ and $\varphi'$ are considered to be
equivalent if they have the same Taylor development of order $r$
at $x\in X$ in some (hence any) pair of coordinate charts centered
at $x$ and $\varphi(x)$. It follows from the definition that each
section $\varphi$ of $Y$ induces the section $j^r\varphi$ of
$J^rY$ by the rule $(j^r\varphi)(x)=j^r_x\varphi$. If $Y$ is
coordinatized by the numbers $\{x^\mu, u^i\}$, where $\{x^\mu\}$
and $\{u^i\}$ are local coordinates in $X$ and $U$, then $(x^\mu,
u^i, u^i_{\mu_1}, u^i_{\mu_1\mu_2}, \ldots, u^i_{\mu_1\cdots
\mu_{r}})$ are local coordinates in $J^rY$ and the section
$j^r\varphi$ is given in these coordinates by
\begin{equation}
x \mapsto (x, \varphi^i(x),
\partial_{\mu_1}\varphi^i(x),\partial_{\mu_1}\partial_{\mu_2}\varphi^i(x),\ldots,\partial_{\mu_1}\ldots\partial_{\mu_r}\varphi^i
(x))\,.
\end{equation}

Classical dynamics on $Y\rightarrow X$ are specified by imposing
differential equations. An $r$-th order differential equation is,
by definition, a closed imbedded submanifold $S \subset J^rY$. A
solution, or \textit{true history} $\varphi$, is a section of $Y$
satisfying $j^r\varphi \in S$. The true histories form a subspace
$\Sigma$ in the space of all histories $M$. In classical field
theory, the differential equations are mostly defined by
(nonlinear) differential operators associated to vector bundles.
Let $E\rightarrow X$ be a vector bundle. The $r$-th order
$E$-valued differential operator is a map $T: M\rightarrow
\Gamma(E)$ such that the value of the section $T(\varphi)$  at
$x\in X$ is fully determined by  $j^r_x\varphi$. This definition
assumes the existence of a bundle morphism $\Theta:
J^rY\rightarrow E$ such that $T(\varphi)=\Gamma(\Theta)\circ
j^r(\varphi)$. The differential equation $S\in J^rY$ is then
identified with the pre-image of the base $X\subset E$, considered
as the zero section, under the bundle map $\Theta$. In each local
coordinate chart on $E\oplus J^rY$, the condition that $\varphi$
belongs to the kernel of the map $T$ takes the form of partial
differential equations
\begin{equation}\label{}
    T_a(x, \varphi^i,
\partial_{\mu_1}\varphi^i,\partial_{\mu_1}\partial_{\mu_2}\varphi^i,\ldots,\partial_{\mu_1}\ldots\partial_{\mu_r}\varphi^i
)=0\,,
\end{equation}
where $\{T_a\}$ are components of the section $T(\varphi)$ with
respect to some frame $\{e^a\}$ in $E$.

Though rigor and geometrically transparent, the language of jets
appears to be unduly cumbersome in developing general
field-theoretical constructions where the local structure of
fields, while important, is out of focus. For this reason we will
use a little bit loose but much more handy notation known in the
physical literature as De Witt's condensed notation \cite{DW}.
According to this notation, the fields $\varphi^i(x)$ are treated
as local coordinates on the infinite dimensional manifold $M$; in
so doing,  the index $i$ also includes the local coordinates
$\{x^\mu\}$ on $X$ so that $\varphi^i\equiv \varphi^i(x)$.
Similarly, the sections $s =s_a(x)e^a$ of the vector bundle
$E\rightarrow X$ are considered to be linear coordinates
$s_a\equiv s_a(x)$ on the infinite dimensional vector space
$\Gamma(E)$. Then  one can view the differential operator $T:
M\rightarrow \Gamma(E)$ as a section of the trivial, infinite
dimensional vector bundle $\mathcal{E}=M\times \Gamma(E)$ over
$M$.  Of a particular importance for our consideration will be
differential operators with values in differential forms on $X$,
this is the case of $E=\wedge^\bullet T^\ast X$. We will denote
the space of all such operators by $\Omega^\bullet$. The scalar
differential operators are, by definition, the elements of the
space $\Omega^0$. Let $\nu$ be a fixed volume form on $X$. By
smooth functions on $M$ we understand the local functionals of
fields. These are given by integrals
\begin{equation}\label{}
   f[\varphi]=\int_X \nu F(\varphi)
\end{equation}
of scalar differential operators $F$ evaluated for fields
$\varphi\in M$ of compact support.  The relation between local
functionals and scalar differential operators is not generally
one-to-one. For example, if $\partial X=\varnothing$ and $F$ and
$F'$ are two scalar operators such that $\nu(F'-F)=dj$ for  some
$j\in \Omega^{n-1}$, $n=\mathrm{dim} X$, then $f'=f$. We say that
two local functionals $f$ and $f'$ are equivalent modulo boundary
terms, $f\doteq f'$, if they differ by a local functional of the
form $\int_X dj(\varphi)$. By the Stokes theorem the equivalence
implies that $f'-f=\int_{\partial X} j(\varphi)$.

In the sequel, for the sake of simplicity, we will assume that
$Y\rightarrow X$ is a trivial vector bundle over a contractible
domain in $\mathbb{R}^n$. Then, one has a simple criterion for the
local functional $f$ to vanish modulo boundary terms \cite{HT},
\cite{BBH}, \cite{Olv}:
\begin{equation}\label{}
     \mathrm{d}f\equiv \int_X\nu \frac{\delta f}{\delta
    \varphi^i}\delta\varphi^i=0   \quad \forall \varphi\quad \mbox{iff} \quad  f\doteq 0\,.
\end{equation}
Of course, in a topologically nontrivial situation only the ``if''
part of the criterion holds true. Following the finite dimensional
pattern, we can  interpret the variations $\mathrm{d}f$ of local
functionals as exact 1-forms on $M$, with $\{\delta \varphi^i\}$
being a local frame in $T^\ast M$. The components of the 1-form
$\mathrm{d}f$ define the Euler-Lagrange differential operator
$\delta f/\delta\varphi^i$.  In a similar manner, one can
introduce the tangent bundle $TM$ as the vector bundle which
sections are \textit{evolutionary vector fields}
\begin{equation}\label{Ev}
    V=\int_X \nu
    V^i(x,\varphi,\ldots,\partial_{\mu_1}\ldots\partial_{\mu_r}
    \varphi)\frac{\delta}{\delta \varphi^i(x)}\,.
\end{equation}
Every evolutionary vector field $V$ generates a  flow $\Phi_t^V$
on the space of all histories $M$, which is described by the
evolution-type system $\partial_t\varphi^i=V^i$ of PDEs, hence the
name.

Given a fiber bundle $Y\rightarrow X$, the correspondence
$E\mapsto \mathcal{E}$ is quite natural in the sense that it
allows one to extend all the usual tensor operations from the
finite dimensional vector bundles to the corresponding infinite
dimensional ones. In particular, one can define the dual of the
vector bundle $\mathcal{E}$ as the vector bundle
$\mathcal{E}^\ast$ corresponding to $E^\ast$. There is a natural
pairing between sections of $\mathcal{E}$ and $\mathcal{E}^\ast$
defined by the integral
\begin{equation}\label{}
    \langle S, T \rangle =\int_X \nu \langle
    S(\varphi),T(\varphi)\rangle\qquad \forall S\in
    \Gamma(\mathcal{E}^\ast)\,, \forall T\in \Gamma(\mathcal{E})\,.
\end{equation}
As is seen the result of pairing is a smooth function on $M$. In
accordance with De Witt's convention, we can write the last
expression just as the sum $S^aT_a$, where the repeated index $a$
implies also integration over $X$.

Each homomorphism $h:E\rightarrow E'$ induces the homomorphism
${H}: \mathcal{E}\rightarrow \mathcal{E}'$ of the associated
vector bundles through the action  on their fibers, $\Gamma(h):
\Gamma(E)\rightarrow \Gamma(E')$. But the infinite dimensional
vector bundles $\mathcal{E}\rightarrow M$  admit more general
homomorphisms. By definition, a general homomorphism $H\in
\mathrm{Hom}(\mathcal{E}, \mathcal{E}')$ is given by a
differential operator
\begin{equation}\label{}
    H: \Gamma(Y\oplus E)\rightarrow \Gamma(E')
\end{equation}
such that the section $H(\varphi, s)\in \Gamma(E')$, where
$\varphi\in M$ and $s\in \Gamma(E)$, depends $\mathbb{R}$-linearly
on the second argument. As usual the homomorphism $H$ induces the
homomorphism on sections. Namely, $\Gamma(H)$ takes a section
$(\varphi, T(\varphi))\in \Gamma(Y\oplus E)$ of $\mathcal{E}$ to
the section $H(\varphi, T(\varphi))$ of $\mathcal{E}'$.

We define the transpose of a homomorphism $H\in
\mathrm{Hom}(\mathcal{E}_1,\mathcal{E}_2)$ as a unique
homomorphism $H^\ast\in
\mathrm{Hom}(\mathcal{E}_2^\ast,\mathcal{E}_1^\ast)$ satisfying
the property
\begin{equation}\label{}
    \langle \Gamma(H)(W),P\rangle \dot =\langle
    W,\Gamma(H^\ast)(P)\rangle
\end{equation}
for all  $W\in
    \Gamma(\mathcal{E}_1)$ and $P\in \Gamma(\mathcal{E}_2^\ast)$.

The direct product $\mathcal{E}_1\oplus \mathcal{E}_2$ is defined
to be the bundle corresponding to $E_1\oplus E_2$.
 Finally, we define the tensor product $$\mathcal{E}_0\otimes \mathcal{E}_1\otimes\cdots\otimes
\mathcal{E}_m$$ as the trivial vector bundle over $M$ whose
sections are differential operators
\begin{equation}
H:\Gamma (Y\oplus E^\ast_1\oplus\cdots \oplus E^\ast_m)\rightarrow
\Gamma(E_0)
\end{equation}
that are linear in the sections of $\Gamma(E^\ast_k)$,
$k=1,\ldots,m$. In case $m=1$, we have the usual isomorphism
$\Gamma(\mathcal{E}_0\otimes \mathcal{E}_1)\simeq
\mathrm{Hom}(\mathcal{E}^\ast_1, \mathcal{E}_0)$.  Applying this
construction to the cotangent bundle $T^\ast M$, we define the
$p$-th exterior power of the cotangent  bundle $\wedge^p T^\ast
M$, whose sections are $p$-forms on $M$. The operator of
variational derivative mentioned above gives rise to the exterior
differential $\mathrm{d}: \Gamma(\wedge^{p}T^\ast M)\rightarrow
\Gamma(\wedge^{p+1}T^\ast M)$ with the property $\mathrm{d}^2=0$.

One can view the theory of jet bundles as a ``differential
geometry with higher derivatives''.  The advantage of the
condensed notation over jets is that it brings the subject back
into the more familiar geometric framework  without higher
derivatives at the cost of passing to infinite dimensional
manifolds. Of course, care must be exercised when using the
standard differential-geometric constructions  in the infinite
dimensional setting. It is particularly important to keep in mind
that unlike the finite dimensional case, the space of sections
$\Gamma(\mathcal{E})$ is not a module over the space of smooth
functions $C^\infty(M)$ (the latter consists of local
functionals). Although it is possible to think of
$\Gamma(\mathcal{E})$ as a module over the scalar differential
operators $\Omega^0$ with respect to the pointwise multiplication
of functions on $X$, the induced action of $\mathrm{Hom}
(\mathcal{E},\mathcal{E}')$ on $\Gamma(\mathcal{E})$ is only
$\mathbb{R}$-linear, not $\Omega^0$-linear. Finally,
$\Gamma(\mathcal{E}_1\otimes \mathcal{E}_2)\neq
\Gamma(\mathcal{E}_1)\otimes \Gamma(\mathcal{E}_2)$.

\section{Symmetries, identities and conservation laws}

The discussion of the previous section can be summarized by saying
that  the classical dynamics of fields are fully specified by a
section $T$ of some vector bundle $\mathcal{E}\rightarrow M$ over
the space of all histories. For this reason we call $\mathcal{E}$
the \textit{dynamics bundle}. The subspace of true histories
$\Sigma\subset M$ is defined to be the zero locus of $T$,
\begin{equation}\label{}
\Sigma=\{\varphi\in M| \;T(\varphi)=0\}\,.
\end{equation}
Using the physical terminology, we will refer to $\Sigma$ as the
\textit{shell}. The field equations $T(\varphi)=0$ are supposed to
satisfy the standard regularity conditions usually assumed for
PDEs, see e.g. \cite{HT}, \cite{BBH}, \cite{Saund}. These
conditions ensure that any section $S\in \Gamma(\mathcal{E}')$
vanishing on $\Sigma$ is representable in the form
\begin{equation}\label{RC}
    S=\Gamma(H)(T)
\end{equation}
for some homomorphism $H\in
\mathrm{Hom}(\mathcal{E},\mathcal{E}')$.

The present section contains no original results, we just explain
our terminology and briefly run through the notions of symmetry,
identities and conservation laws, which are studied in the next
sections.

To begin with, we introduce the homomorphism $J: TM\rightarrow
\mathcal{E}$ defined by\footnote{By abuse of notation, we will use
the same symbol to denote a vector bundle homomorphism and the
induced  homomorphism on sections. In particular, we will write
$J$ instead of more pedantic $\Gamma(J)$.}
\begin{equation}
\langle J(X),\Psi\rangle=X\cdot \langle T , \Psi\rangle\qquad
\forall X\in \Gamma(TM),\quad  \forall \Psi\in \Gamma(E^\ast)\,,
\end{equation}
$\Psi$ being understood as a field independent section of
$\mathcal{E}^\ast$. The homomorphism $J$ is called \cite{KLV} the
\textit{universal linearization} of the nonlinear differential
operator $T$. We also introduce the sign of weak equality
$S\approx S'$ borrowed from Dirac's constrained dynamics. It means
that two sections $S$ and $S'$ of some vector bundle over $M$
coincide \textit{on shell}, i.e.,  $S|_\Sigma =S'|_\Sigma$.

An evolutionary vector field $X$ is said to be a \textit{symmetry}
of the classical system if it preserves the shell $\Sigma$. This
amounts to the weak equality
\begin{equation}\label{sym}
    J(X)\approx 0\,.
\end{equation}
The symmetries form a subalgebra $\mathrm{Sym}'(T)\subset
\Gamma(TM)$ in the Lie algebra of all vector fields on $M$.

A section $\Psi\in \Gamma(\mathcal{E}^\ast)$ is said to be an
\textit{identity} of the classical system if
\begin{equation}\label{id}
\langle\Psi,T\rangle \doteq  0\,.
\end{equation}
 As a differential consequence of the last equality we have
\begin{equation}\label{id'}
J^\ast(\Psi)\approx 0\,.
\end{equation}
Denote by $\mathrm{Id}'(T)\subset \Gamma(\mathcal{E}^\ast)$ the
subspace of all identities.

 The notions of symmetry and identity can be further elaborated on.
Observe that any vector field $X$ that vanishes on the shell,
i.e., $X\approx 0$, is a symmetry in the sense of (\ref{sym}). The
on-shell vanishing symmetries are called \textit{trivial}. They
are present in any classical theory, containing no valuable
information about the structure of dynamics. Fortunately, being
proportional to the equations of motion, the trivial symmetries
constitute an ideal $\mathrm{Sym}_0(T)\subset \mathrm{Sym}'(T)$ in
the Lie algebra of all symmetries and can thus be systematically
disregarded. So in what follows, by a symmetry we will actually
mean  an element of the factor algebra
$\mathrm{Sym}(T)=\mathrm{Sym}'(T)/\mathrm{Sym}_0(T)$.

A symmetry $X\in \mathrm{Sym}(T)$ is called a \textit{gauge
symmetry}  if there exists  a vector bundle
$\mathcal{F}\rightarrow M$ together with a section $\varepsilon\in
\Gamma(\mathcal{F})$ and a homomorphism $R\in \mathrm{Hom}
(\mathcal{F}, TM)$ such that $\mathrm{Im} R \subset
\mathrm{Sym}(T)$  and $X=R(\varepsilon)$. If the vector bundle
$\mathcal{F}$ is big enough to accommodate  any (nontrivial) gauge
symmetry with a fixed $R$, then we refer to  $\mathcal{F}$ as the
\textit{gauge algebra bundle}. Notice that $J\circ R\approx 0$.
The homomorphism $R$ defines an (over)complete basis of gauge
generators. Clearly, the vector distribution $\mathrm{Im}R \subset
TM$ is on-shell involutive, and hence it defines a foliation of
$\Sigma$. The leaves  of this foliation are called \textit{gauge
orbits}. Two points $p,q\in\Sigma$ are considered to be (gauge)
equivalent, $q\sim p$, if they belong to the same gauge orbit. The
quotient space $\Sigma/\sim$ is termed as a \textit{covariant
phase space}. It is the space of all physical states of the
system.

Since the gauge symmetries form an ideal  $\mathrm{GSym}(T)\subset
\mathrm{Sym}(T)$, one can define the factor algebra
$\mathrm{RSym}(T)=\mathrm{Sym}(T)/\mathrm{GSym(T)}$, which is
naturally identified with the Lie algebra of all \textit{rigid
symmetries}. By definition, the rigid symmetries are transverse to
the gauge orbits and thus they induce a nontrivial action on the
phase space of physical states.

Equation (\ref{id}) admits a lot of trivial solutions that can be
constructed as follows. Take an arbitrary section of $K\in
\Gamma(\mathcal{E}^\ast\wedge \mathcal{E}^\ast)$ and consider it
as a homomorphism from
$\mathrm{Hom}(\mathcal{E},\mathcal{E}^\ast)$. Then $\Psi=K(T)$
satisfies (\ref{id}). The  identities constructed in such a way
are of no physical significance  and, therefore, they should be
regarded as trivial. Notice that each trivial identity vanishes on
shell. The converse implication is also true \cite{HT}: Each
on-shell vanishing identity has the form $\Psi=K(T)$ for some
skew-symmetric $K$. The proof exploits the regularity assumption
for the equations of motion. The nontrivial identities are then
defined as the elements of the quotient space
$\mathrm{Id}(T)=\mathrm{Id}'(T)/\mathrm{Id}_0(T)$, where
$\mathrm{Id}_0(T)\subset \mathrm{Id}(T)$ is the subspace of
trivial identities.

An identity $\Psi\in \mathrm{Id}(T)$ is called  \textit{Noether's
identity} if there exists a vector bundle $\mathcal{G}\rightarrow
M$ together with a section $\xi\in \Gamma(\mathcal{G})$ and a
homomorphism $Z\in \mathrm{Hom}(\mathcal{G}, \mathcal{E}^\ast)$
such that $\mathrm{Im} Z\subset \mathrm{Id}(T)$ and $\Psi=Z(\xi)$.
Denote by $\mathrm{NId}(T)\subset \mathrm{Id}(T)$ the subspace of
all Noether's identities and define the quotient space
$\mathrm{Char}(T)=\mathrm{Id}(T)/\mathrm{NId}(T)$, whose elements
are called \textit{characteristics}. Again, one can choose the
vector bundle $\mathcal{G}$ to be big enough so that
$\mathrm{NId}(T)=\mathrm{Im} Z$. In this case, $Z$ is said to
define an (over)complete basis of the generators of  Noether's
identities. Notice that $Z^\ast(T)=0$, where $Z^\ast$ is transpose
of $Z$. It is quite natural to call $\mathcal{G}$ the
\textit{bundle of the Noether identities}.

All we have said above can be concisely reformulated in terms of
the following sequence of homomorphisms
\begin{equation}\label{4t-seq}
    \xymatrix{0\ar[r]&
    \Gamma(\mathcal{F})\ar[r]^-{R}&\Gamma(TM)
    \ar[r]^-{J}&\Gamma(\mathcal{E})\ar[r]^{Z^\ast}
    &\Gamma(\mathcal{G}^\ast)\ar[r]&
    0}
\end{equation}
and its transpose
\begin{equation}\label{4t-seq*}
\xymatrix{
    0&
    \Gamma(\mathcal{F}^\ast)\ar[l]&\Gamma(T^\ast
    M)\ar[l]_{R^\ast}&\Gamma(\mathcal{E}^\ast)\ar[l]_-{J^\ast}
    &\Gamma(\mathcal{G})\ar[l]_-{Z}&
    0\ar[l]}.
\end{equation}
Upon restriction to $\Sigma$ these sequences make cochain
complexes; the properties $Z^\ast\circ J\approx 0$ and $
J^\ast\circ Z\approx 0$ follow from the differential consequence
of the identity $Z^\ast(T)=0$. The spaces of all rigid symmetries
and characteristics are then naturally identified with the
cohomology groups
\begin{equation}\label{}
    \mathrm{RSym}(T) \simeq \frac{\mathrm{Ker}
    J|_\Sigma}{\mathrm{Im}R|_\Sigma}\,,\qquad
    \mathrm{Char}(T) \simeq\frac{\mathrm{Ker} J^\ast|_\Sigma}{\mathrm{Im}
    Z|_{\Sigma}}\,.
\end{equation}

Intimately  related  to the notion of an identity is the notion of
a \textit{conservation law}. The latter  is identified with a
form-valued differential operator $j\in \Omega^{n-1}$ taking true
histories to closed $(n-1)$-forms,
\begin{equation}\label{}
    dj\approx 0\,.
\end{equation}
In view of the regularity assumptions, the $n$-form $dj$ has to be
proportional to the equations of motion, that is
\begin{equation}\label{icl}
    \int_X dj=\langle\Psi,T\rangle
\end{equation}
for some $\Psi\in \Gamma(\mathcal{E}^\ast)$. Comparing the last
equality with (\ref{id}), we see that every on-shell closed
$(n-1)$-form gives rise to some (perhaps trivial) identity and
vice versa. The relation
\begin{equation}
\label{rel} \mbox{(Identities)}\leftrightarrow\mbox{(Conservation
Laws)}
\end{equation}
is far from being  a bijection, since neither left nor right hand
sides of (\ref{icl}) are defined by $j$ and $\Psi$ unambiguously.
For instance, one may add to $j$ any closed (and hence exact) form
$di$ without changing $\Psi$. The exact forms $di\in \Omega^{n-1}$
are characterized by zero charge and, therefore, one should regard
them as trivial. Another source of triviality is related to the
currents $j$ that vanish on shell. Taking into account either  of
possibilities, we identify the space of nontrivial conservation
laws $\mathrm{CL}(T)$ with the equivalence classes of on-shell
closed forms $j\in \Omega^{n-1}$. Two such forms are considered as
equivalent whenever they differ by an on-shell exact
form\footnote{This definition does not exclude the possibility
that the charge, being related to a nontrivial conserved current,
vanishes identically. For discussion of this phenomenon, see
\cite{BH}.}:
\begin{equation}\label{}
j\sim j'\quad \Leftrightarrow \quad j-j'\approx di\,.
\end{equation}
The space $\mathrm{CL}(T)$ is also known as the
\textit{characteristic cohomology} of $\Sigma$ in form degree $n-1$,
see  \cite{BrGr}, \cite{Vin} and \cite{HKS}.

Turning to the right hand side of (\ref{icl}), one can see that
trivial identities give rise to trivial conservation laws. This is
an additional (in fact the main) reason to call the on-shell
vanishing identities trivial. Furthermore, the Noether identities
correspond to the trivial conservation laws as well.

Considering (\ref{rel}) modulo trivialities, one arrives at the
following relation
\begin{equation}\label{}
\mathrm{Char}(T)\leftrightarrow\mathrm{CL}(T) \,.
\end{equation}
This last relation is in essence an isomorphism  \cite{Olv},
\cite{BBH}, so that the problem of constructing nontrivial
conservation lows boils down to finding out characteristics.

\vspace{4mm} \noindent \textbf{Example}. Generally it is a rather
hard problem to identify  the spaces of all non-trivial symmetries
and characteristics for a given set of equations. The problem,
however, is considerably simplified for the ordinary differential
equations in normal form, where either of spaces admits a fairly
explicit description. By this reason we will  use this class of
dynamical systems to exemplify all the general notions and
constructions throughout the paper. Without loss of generality we
can restrict ourselves  to the systems of  first-order ODEs
\begin{equation}\label{ODE}
    T^i\equiv \dot x^i(t)+v^i(t,x)=0\,.
\end{equation}
Here the overdot stands for the derivative with respect to the
independent variable $t\in \mathbb{R}$ and $\{x^i\}$ are local
coordinates on $Y$.  Considering  $Y\times \mathbb{R}$ as a
trivial fiber bundle over $\mathbb{R}$, one can thought of $v$ as
a vertical vector field on $Y\times \mathbb{R}$. Equations
(\ref{ODE}) are neither reducible nor gauge invariant and the
dynamics bundle is naturally identified with the tangent bundle
$TM$ of the trajectory space $M=\mathrm{Map}(\mathbb{R},Y)$. By
definition, the characteristics of (\ref{ODE}) live in the space
$\Gamma(T^\ast M)$ spanned by 1-forms
\begin{equation}\label{l1f}
    \Psi=\int dt \psi_i(t, x,\dot x,\ldots,\stackrel{_{(m)}}{x})\delta
    x^i(t)\,.
\end{equation}
The system  (\ref{ODE}) being regular, one can exclude all the
derivatives from the integrand of (\ref{l1f}) by means of the
equations of motion, looking for characteristics of the form
\begin{equation}\label{psi}
    \Psi=\int dt \psi_i(t,x)\delta x^i(t)\,,
\end{equation}
where $\psi=\psi_i(t, x)dx^i$ is a vertical 1-form on $Y\times
\mathbb{R}$. Relation (\ref{id}) results in two conditions
\begin{equation}\label{vf}
\psi=\widetilde{d}f\,,\qquad  \partial_t f=v\cdot f\,,
\end{equation}
where  $f$ is a function on $Y\times \mathbb{R}$ and
$\widetilde{d}$ is the exterior differential on $Y$. Let
$C_v^\infty(Y\times \mathbb{R})$ denote the space of all smooth
functions satisfying equations (\ref{vf}). These equations specify
$f$ up to an additive constant and each such $f$ is conserved.
Thus we are led to the following isomorphisms:
\begin{equation}\label{}
    \mathrm{Char}(\dot x+v)\simeq C^\infty_v(Y\times \mathbb{R})/\mathbb{R}\,,\qquad \mathrm{CL}(\dot x+v)\simeq C^{\infty}_v(Y\times \mathbb{R})\,.
\end{equation}

Let us now turn to the symmetries of  (\ref{ODE}). The general
evolutionary  vector field on $M$ has the form
\begin{equation}\label{W}
    W=\int dt w^i(t,x,\dot x,\ldots,\stackrel{_{(m)}}{x})\frac{\delta}{\delta
    x^i(t)}\,.
\end{equation}
Being  interested in symmetries, we can exclude all the
derivatives from the integrand of (\ref{W}) with the help of the
equations of motion. This reduces the ansatz (\ref{W}) to
\begin{equation}\label{}
    W=\int dt w^i(t,x)\frac{\delta}{\delta x^i(t)}\,,
\end{equation}
where $w=w^i(t,x)\partial/\partial x^i$ is a vertical vector field
on $Y\times \mathbb{R}$. Verifying (\ref{sym}) one can find that
$W$ is a symmetry iff
\begin{equation}\label{vw}
    \partial_t w=[v,w]\,.
\end{equation}
Denoting by $\frak X_v(Y\times \mathbb{R})$ the space of all
vertical vector fields satisfying (\ref{vw}), we can write
\begin{equation}\label{}
    \mathrm{Sym}(\dot x+v)=\frak X_v(Y\times \mathbb{R})\,.
\end{equation}
Let us suppose that the vector field $v$ is complete. Then,
considering (\ref{vf}) and (\ref{vw}) as systems of the
first-order ODEs for $f$ and $w$, respectively, we conclude that
the spaces of conservation laws and symmetries are isomorphic to
the spaces of initial data $C^{\infty}(Y)$ and $\frak{X}(Y)$. Thus
the space of symmetries  appears to be  ``much larger'' than the
spaces of characteristics and conservation laws whenever $\dim
Y>1$.

\vspace{4mm}

The form of equations  (\ref{sym}) and (\ref{id'}) shows a certain
duality between the concepts of symmetry and conservation law. The
dual nature of these concepts is also supported by  the celebrated
Noether theorem for Lagrangian equations of motion. The long time
use of this theorem has lead to a widespread belief that any
conservation law or, what is the same, characteristic is a
reflection of some symmetry. Beyond the scope of Lagrangian
dynamics, this is not always true. Not any conservation law comes
from some symmetry, nor does each symmetry correspond to a
conservation law. There is no canonical relationship  between
these two concepts in general. This fact has a simple geometric
explanation: The characteristics $\Psi$'s and the symmetries $X$'s
belong to different vector bundles unless $\mathcal{E}^\ast=TM$.
This suggests that (i) any systematic procedure for converting
symmetries to characteristics and vice versa  has to involve an
additional geometric structure on $M$ and (ii) whatever this
structure may be, it will result in a linear relation between the
dual of the dynamics bundle $\mathcal{E}^\ast$ and the tangent
bundle of $M$. It is then quite reasonable to start with some
homomorphism $V: \mathcal{E}^\ast\rightarrow TM$ and to look for a
suitable set of conditions ensuring  that $V(\Psi)$ is a symmetry
whenever $\Psi$ is a characteristic. Actually, such a homomorphism
$V$ has been already identified as a key structure needed for the
path-integral quantization of non-Lagrangian gauge theories
\cite{KLS}. This homomorphism is called a \textit{Lagrange
anchor}.

\section{The Lagrange anchor}
\begin{defn}\label{def1}
A vector bundle homomorphism $V: \mathcal{E}^\ast\rightarrow TM$
is called a Lagrange anchor, if the diagram
\begin{equation}\label{DV}
    \xymatrix{\Gamma(TM)\ar[r]^J&\Gamma(\mathcal{E})\\
\Gamma(\mathcal{E}^\ast)\ar[r]^{J^\ast} \ar[u]_{V}&\Gamma(T^\ast
M)\ar[u]_{V^\ast}
    }
\end{equation}
commutes upon restriction to the shell, i.e.,
\begin{equation}\label{JV}
    J\circ V\approx V^\ast\circ J^\ast\,.
\end{equation}
\end{defn}
Since the last relation is linear in $V$, the Lagrange anchors
form a vector space, which we denote by  $\mathrm{An}'(T)$. This
space is clearly nonempty as each on-shell vanishing homomorphism
$V$ obeys (\ref{JV}). The on-shell vanishing Lagrange anchors are
of no significance from the perspective of converting conservation
laws to rigid symmetries and we rule them out by passing to the
quotient $\mathrm{An}(T)=\mathrm{An}'(T)/\mathrm{An}_0(T)$, where
$\mathrm{An}_0(T)$ is the subspace of all on-shell vanishing
Lagrange anchors.
\begin{prop}\label{prop1}
The Lagrange anchor takes identities to symmetries, trivial
identities to trivial symmetries  and the Noether identities to
the gauge symmetries.
\end{prop}
\begin{proof}
If $\Psi$ is an identity, then $J^\ast(\Psi)\approx 0$ and
\begin{equation}\label{}
J\circ V(\Psi)\approx V^\ast\circ J^\ast(\Psi)\approx 0\,.
\end{equation}
Hence $X=V(\Psi)$ is a symmetry. It is also clear that $X\approx
0$ whenever $\Psi\approx 0 $. If now $\Psi=Z(\xi)$ for some
$\xi\in \Gamma(\mathcal{G})$, then
\begin{equation}\label{}
    J\circ V(\Psi)\approx V^\ast\circ J^\ast\circ Z(\xi)\approx
    0\qquad \forall \xi \in \Gamma(\mathcal{G})\,.
\end{equation}
Thus, $X=V\circ Z(\xi)$ is a gauge symmetry. Since the generators
$R$ are assumed to form an (over)complete basis in the gauge
algebra bundle, there exists a homomorphism $W:
\mathcal{G}\rightarrow \mathcal{F}$ such that the diagram
\begin{equation}\label{DW}
    \xymatrix{\Gamma(\mathcal{F})\ar[r]^R&\Gamma(TM)\\
\Gamma(\mathcal{G})\ar[r]^{Z}
\ar[u]_W&\Gamma(\mathcal{E}^\ast)\ar[u]_V
    }
\end{equation}
commutes upon restriction to the shell, i.e., $V\circ Z\approx
R\circ W$.
\end{proof}
Let us  combine all the diagrams (\ref{4t-seq}), (\ref{4t-seq*}),
(\ref{DV}), (\ref{DW}) into the following one:
\begin{equation}\label{}
    \xymatrix {   0 \ar[r]   & \Gamma(\mathcal{F}) \ar[r]^{R}&\Gamma(TM) \ar[r]^{J}& \Gamma(\mathcal{E})\ar[r]^{Z^\ast}&
    \Gamma(\mathcal{G}^\ast)\ar[r]&0
    \\
   0 \ar[r]   & \Gamma(\mathcal{G})\ar[u]_{W} \ar[r]^{Z}&\Gamma(\mathcal{E}^\ast)\ar[u]_{V} \ar[r]^{J^\ast}&\Gamma
   ( T^\ast M)\ar[u]_{V^\ast}\ar[r]^{R^\ast}& \Gamma(\mathcal{F}^\ast)\ar[u]_{W^\ast}\ar[r]&0
   }
\end{equation}
From the previous section we know that the rows of this diagram
make cochain complexes upon restriction to the  shell. Then, the
on-shell commutativity of the squares implies that the upward
arrows define a cochine map. It is the standard fact of
homological algebra that each cochain map induces a well defined
homomorphism in cohomology. In the case at hand this gives the
homomorphism
\begin{equation}\label{HV}
    H(V): \mathrm{Char}(T)\rightarrow \mathrm{RSym}(T)
\end{equation}
from the space of characteristics to the space of rigid
symmetries.

A natural question arises about identifying the different anchors
from $\mathrm{An}(T)$  which induce the same homomorphism
(\ref{HV}). An appropriate algebraic concept for examining   the
question is that of homotopy. We say that two anchors $V,
\tilde{V}\in \mathrm{An}(T)$ are  equivalent, $V\sim \tilde{V}$,
if the corresponding cochain maps are homotopic. The latter
implies the existence of homomorphisms $G$ and $K$ such that
\begin{equation}\label{V-V}
\begin{array}{ll}
   \tilde{V}-V\approx G\circ J^\ast+R\circ K\,,&\qquad \tilde{V}{}^\ast-V^\ast\approx J\circ G+K^\ast\circ R^\ast\,,
  \\[3mm]
 \tilde{W}-W\approx K\circ Z\,,&\qquad \tilde{W}{}^\ast-W^\ast\approx Z^\ast\circ
 K^\ast\,.
  \end{array}
\end{equation}
All the relevant maps  are depicted in the following diagram:
\begin{equation}\label{}
    \xymatrix {   0 \ar[r]   & \Gamma(\mathcal{F}) \ar[r]^{R}&\Gamma(TM) \ar[r]^{J}& \Gamma(\mathcal{E})\ar[r]^{Z^\ast}&
    \Gamma(\mathcal{G})\ar[r]&0
    \\
   0 \ar[r]   & \Gamma(\mathcal{G}^\ast)\ar[ul] \ar@<0.5ex>[u]^{\tilde{W}}\ar@<-0.5ex>[u]_{W} \ar[r]^{Z}&\Gamma(\mathcal{E}^\ast)\ar[ul]|{K}
   \ar@<0.5ex>[u]^{\tilde{V}}\ar@<-0.5ex>[u]_{V}
    \ar[r]^{J^\ast}&\Gamma
   ( T^\ast M)\ar@{.>}[ull]\ar@{->}[ul]|{G}\ar@<0.5ex>[u]^{\tilde{V}{}^\ast}\ar@<-0.5ex>[u]_{V^\ast}\ar[r]^{R^\ast}&
   \Gamma(\mathcal{F}^\ast)\ar@{.>}[ull]\ar[ul]|{K^\ast}\ar@<0.5ex>[u]^{\tilde{W}{}^\ast}\ar@<-0.5ex>[u]_{W^\ast}\ar[r]&0\ar[ul]
   }
\end{equation}
The weak equalities  in the first line of
(\ref{V-V})\footnote{These equivalence relations for the Lagrange
anchors originally appeared in  \cite{LS2}. As usual, the homotopy
equivalence classes of cochain  transformations can be
equivalently reinterpreted as a cohomology  classes of
corresponding cochain complex \cite{KLS}. In the context of local
field theory  the equivalence classes of Lagrange anchors were
recently discussed in  \cite{BG}.} imply that
\begin{equation}\label{JG}
    J\circ (G-G^\ast)\approx 0\,.
\end{equation}
This is the only condition on the homomorphisms $G$ and $K$.
Taking into account the fact of completeness of the gauge algebra
generators, we can rewrite (\ref{JG}) in the following equivalent
form:
\begin{equation}\label{}
    G-G^\ast\approx R\circ U
\end{equation}
for some $U: T^\ast M \rightarrow \mathcal{F}$ satisfying
\begin{equation}\label{}
    R\circ U+U^\ast\circ R^\ast\approx 0\,.
\end{equation}
In the diagram above, the homomorphisms $U$ and $-U^\ast$ are
depicted  by the dotted arrows, so that the corresponding
parallelogram of maps appears to be on-shell commutative. Since
the homotopic cochain maps are known to induce the same
homomorphism in cohomology, we conclude that $H(V)=H(\tilde{V})$
whenever the anchors $V$ and $\tilde{V}$ are equivalent.

Denote by $\mathrm{HAn}(T)$ the space of the homotopy classes of
Lagrange anchors. This can be thought of as the quotient of
$\mathrm{An}'(T)$ by the Lagrange anchors of the form
\begin{equation}\label{Vtriv}
V_{\mathrm{triv}}=R\circ K+G\circ J^\ast +V_0\,,
\end{equation}
 where $V_0\approx
0$ and $G\approx G^\ast$. We have
\begin{equation}\label{}
    V_{\mathrm{triv}}(\Psi)=R(\varepsilon)-X_0\,,
\end{equation}
where $\varepsilon=K(\Psi)$ and $X_0=G\circ
J^\ast(\Psi)+V_0(\Psi)\approx 0$. In other words, the trivial
Lagrange anchor $V_{\mathrm{triv}}$ takes each identity (trivial
or not) to a trivial rigid symmetry.

\vspace{4mm}
 \noindent \textbf{Example} (continuation). Consider the dynamical
 system  (\ref{ODE}).
Then the most general local ansatz  for the Lagrange anchor $V:
T^\ast M\rightarrow TM$ looks like
\begin{equation}\label{Vep}
    \langle V(W),P\rangle\dot =\langle W,V^\ast(P)\rangle =\sum_{n=0}^N\int dt w_i(t)\alpha^{ij}_n(t, x,\dot x,\ldots,
    \stackrel{_{(m)}}{x})\stackrel{_{(n)}}{p}_j(t)\,,
\end{equation}
where $W= w_i\delta x^i$ and $P=p_i\delta x^i$ are arbitrary
1-forms on $M$. The homomorphism $J: TM\rightarrow TM$ is given by
the matrix first-order differential operator
\begin{equation}\label{}
    J=\left(\delta^i_j\frac d{dt}+\frac{\partial v^i}{\partial x^j}\right)\delta(t-t')
\end{equation}
The peculiar form of the operator $J$ allows one to absorb all the
derivatives of $p$'s into the trivial anchor (\ref{Vtriv}) for an
appropriate $G$. Using then the equations of motion or, what is
the same, adding an appropriate on-shell vanishing anchor $V_0$,
one can also exclude all the derivatives of $x$'s from the
integrand of (\ref{Vep}). We are led to conclude that each
equivalence class of the Lagrange anchors for equations
(\ref{ODE}) contains the anchor of the form
\begin{equation}\label{Vw}
V(\delta x^i(t))= \alpha^{ij}(t,x)\frac{\delta}{\delta x^j}\,.
\end{equation}
Substituting the last expression to the defining condition
(\ref{JV}), we finally get
\begin{equation}\label{alpha}
    \alpha^{ij}=-\alpha^{ji}\,,\qquad
    \partial_t\alpha=L_v\alpha\,,
\end{equation}
Thus, there is a one-to-one correspondence between the space of
non-equivalent Lagrange anchors $\mathrm{An}(\dot x+v)$ and the
vertical bivector fields
$\alpha=\alpha^{ij}(t,x)\partial_i\wedge\partial_j$ on $Y\times
\mathbb{R}$ satisfying (\ref{alpha}). If the vector field $v$ is
complete, then equation (\ref{alpha}) defines an isomorphism
$\mathrm{An}(\dot x+ v)\simeq \Gamma(\wedge^2 TY)$. This
classification result is in agreement with the conclusions of
paper \cite{BG}. In \cite{BG}, it is proven that any Lagrange
anchor for the AKSZ-type sigma models \cite{AKSZ} always reduces
to a purely algebraic operator whose coefficients do not dependent
on the space-time derivatives of fields.

\section{Proper symmetries and proper deformations}

We begin with an alternative geometric interpretation of the
Lagrange anchor as a ``Lie algebroid with relaxed integrability
condition''. It was first introduced in \cite{KLS}, and termed the
\textit{Lagrange structure}. To fix ideas,  let us think for a
while of $\mathcal{E}\rightarrow M$ as if it were an ordinary
vector bundle over a finite dimensional manifold \footnote{See
comments at the end of Sec. 2.}. The way of  adapting  this
construction to the context of local field theory will be obvious.

\begin{defn}\label{LS} Given a classical system $(\mathcal{E},T)$, a \textit{{Lagrange structure}} is a
homomorphism  $\mathrm{d}_\mathcal{E}: \Gamma(\wedge^\bullet
\mathcal{E})\rightarrow \Gamma(\wedge^{\bullet+1}\mathcal{E})$
obeying two conditions:
\begin{enumerate}
      \item [(i)] \,  $\mathrm{d}_\mathcal{E}$ is a derivation of degree
      1, i.e.,
$$
 \mathrm{d}_{\mathcal{E}}(A\wedge B)=\mathrm{d}_\mathcal{E}A\wedge B +
(-1)^{p}A\wedge \mathrm{d}_\mathcal{E} B\,,
$$
for any $A\in \Gamma(\wedge^p \mathcal{E})$ and $ B\in
\Gamma(\wedge^\bullet \mathcal{E})$;
 \item [(ii)] \,
$\mathrm{d}_\mathcal{E} T=0$\,.
\end{enumerate}
Here we identify $\Gamma(\wedge^0\mathcal{E})$ with $C^\infty(M)$.
\end{defn}

\vspace{5mm} Due to the Leibnitz rule (i), in each trivializing
chart ${U}\subset M$ the operator $\mathrm{d}_\mathcal{E}$ is
completely determined  by its action on the coordinate functions
$\varphi^i$ and the frame sections $\{e^a\}$ of
$\mathcal{E}|_{U}$:
\begin{equation}\label{dE}
    \mathrm{d}_\mathcal{E} \varphi^i=V^i_a(\varphi) e^a\,,\qquad \mathrm{d}_\mathcal{E} e^a=\frac12C_{bc}^a(\varphi)e^b\wedge
    e^c\,.
\end{equation}
Applying $\mathrm{d}_\mathcal{E}$ to the section $T=T_ae^a$, one
can see that the property (ii) is equivalent to the structure
relations
\begin{equation}\label{Brel}
\mathrm{d}_\mathcal{E} T =\frac12(V^i_a\partial_i T_b
-V^i_b\partial_i T_a +C_{ab}^cT_c )e^a\wedge e^b=0\,.
\end{equation}
Clearly, the derivation $\mathrm{d}_\mathcal{E}$ defines a bundle
homomorphism $V: \mathcal{E}^\ast\rightarrow TM$. The shell
$\Sigma$ being a regularly imbedded submanifold, equation
(\ref{Brel}) is equivalent to the defining condition  (\ref{JV}).
Thus, $V$ is a Lagrange anchor.

Dualizing relations (\ref{dE}) one gets a bracket
\begin{equation}\label{LB}
[\,\cdot\,,\,\cdot\,]: \Gamma(\mathcal{E}^\ast)\wedge
\Gamma(\mathcal{E}^\ast)\rightarrow \Gamma(\mathcal{E}^\ast)\,.
\end{equation}
On frame sections $e_a$ of $\mathcal{E}^\ast|_U$ it reads
\begin{equation}
    [e_a,e_b]=C_{ab}^d e_d
\end{equation}
and extends to arbitrary  $ \Psi, \Psi'\in
\Gamma(\mathcal{E}^\ast)$ by the Leibnitz rule
\begin{equation}\label{LB1}
    [f\Psi,\Psi']=f[\Psi,\Psi']+(V(\Psi)\cdot f)\Psi'\,,\qquad \forall
    f\in C^\infty(M)\,.
\end{equation}
Generally  the bracket violates the Jacobi identity, therefore  it
is not a Lie-algebra bracket.

\begin{defn}\label{ILS}
A Lagrange structure $(\mathcal{E},T,\mathrm{d}_\mathcal{E})$ is
said to be integrable if $\mathrm{d}_\mathcal{E}^2=0$.
\end{defn}

Comparing Definitions \ref{LS} and \ref{ILS} with the definition
of a \textit{Lie algebroid} (see e.g. \cite{Mac}, \cite{CdSW}),
one can see that  the integrable Lagrange structure is nothing
else but the Lie algebroid over $M$ with anchor
$V:\mathcal{E}^\ast\rightarrow TM$, Lie bracket (\ref{LB}), and a
distinguished section $T$. This particular yet important case
accounts for the name of $V$ -- ``Lagrange anchor''. For an
integrable Lagrange structure the bracket (\ref{LB}) satisfies the
Jacobi identity and the anchor distribution $\mathrm{Im}V\subset
TM$ is integrable.

\vspace{3mm}
To illustrate (as well as motivate) the definitions above let us
consider a Lagrangian theory with action $S(\varphi)$. The
equations of motion read
\begin{equation}\label{dS}
   T\equiv \mathrm{d}S=0\,,
\end{equation}
so that the dynamics bundle $\mathcal{E}$ is given by the
cotangent bundle $T^\ast M$ of the space of all histories. The
canonical Lagrange structure is given by the exterior differential
$\mathrm{d}: \Gamma(\wedge^\bullet T^\ast M)\rightarrow
\Gamma(\wedge^{\bullet+1}T^\ast M)$ and the defining condition
(\ref{Brel}) takes the form
\begin{equation}\label{ddS}
    \mathrm{d}T=\mathrm{d}{}^2S=0\,.
\end{equation}
The identity $\mathrm{d}^2=0$ means integrability. The Lagrange
anchor is given by the identical map $ V=\mathrm{id}:
TM\rightarrow TM$ and the Lie bracket (\ref{LB}) coincides with
the commutator of vector fields.

Definitions  \ref{LS} and \ref{ILS}   suggest to view the Lagrange
structure as a ``Lie algebroid with relaxed integrability
condition'' in the sense that the ``strong'' integrability
condition $\mathrm{d}_{\mathcal{E}}^2=0$ is replaced here by the
weaker one $\mathrm{d}_{\mathcal{E}}^2T=0$.  In spite of weakened
integrability, it is still possible to utilize many of the usual
differential-geometric constructions associated to Lie algebroids.
In particular, we can endow the exterior algebra of
$\mathcal{E}$-differential forms $\Gamma(\wedge^\bullet
\mathcal{E})$ with the operations of inner differentiation and Lie
derivative.

 By definition, the inner differential $i_\Psi$
associated with a section $\Psi\in \Gamma(\mathcal{E}^\ast)$ is a
differentiation of $\Gamma(\wedge^\bullet \mathcal{E})$ of degree
-1 which action on $A \in \Gamma(\mathcal{E})$ is given by
\begin{equation}\label{}
    i_\Psi A =\langle\Psi,A\rangle\,.
\end{equation}
Following the analogy with Cartan's calculus, the Lie derivative
along $\Psi$ is defined now by
\begin{equation}\label{}
    L_\Psi =i_\Psi\circ \mathrm{d}_{{\mathcal{E}}}+\mathrm{d}_{{\mathcal{E}}}\circ
    i_\Psi\,.
\end{equation}
It is a differentiation of $\Gamma(\wedge^\bullet \mathcal{E})$ of
degree 0. One can easily verify  the following identities:
\begin{equation}\label{Rels}
\begin{array}{ll}
    \{i_\Psi, i_{\Psi'}\}=0\,,\qquad&
    [L_\Psi,\mathrm{d}_\mathcal{E}]=[i_\Psi, \mathrm{d}{}^2_\mathcal{E}]\,,\\[5mm]
[L_\Psi, i_{\Psi'}]=i_{[\Psi,\Psi']}\,,\quad\qquad & [L_{\Psi},
L_{\Psi'}]=L_{[\Psi,\Psi']}+\{i_{\Psi'},[i_\Psi,
\mathrm{d}_{\mathcal{E}}^2]\}\,,
\end{array}
\end{equation}
where the braces stand for anticommutators. We can also define the
action of the Lie derivative on the contravariant
$\mathcal{E}$-tensors by setting
\begin{equation}\label{}
    L_\Psi \Psi' = [\Psi,\Psi'] \qquad \forall \Psi,\Psi' \in
    \Gamma(\mathcal{E}^\ast)\,.
\end{equation}
This definition extends to arbitrary $\mathcal{E}$-tensor fields
by the usual formulas of differential geometry. From this point on
we return to infinite-dimensional setting of Sec. 2.

\begin{defn}\label{prosym}
Given a Lagrange structure $(\mathcal{E}, T,
\mathrm{d}_{\mathcal{E}})$ and a section $\Psi\in
\Gamma(\mathcal{E}^\ast)$, the vector field $V(\Psi)$ is said to
be a proper symmetry generated by $\Psi$ if $L_\Psi T=0$.
\end{defn}

Notice that the proper symmetry $V(\Psi)$ is a symmetry in the
usual sense. Indeed, $L_{\Psi}T\approx J\circ V(\Psi)$ and $L_\Psi
T=0$ implies $V(\Psi)\in \mathrm{Sym}'(T)$. For variational
equations of motion equipped with the canonical Lagrange structure
(\ref{ddS}) the proper symmetries coincide with the symmetries of
action.

We emphasize that the property of being proper symmetry is not
homotopy invariant: If $V(\Psi)$ is a proper symmetry and $V\sim
V'$, then $V'(\Psi)$ may not be a symmetry at all,  unless
$\Psi\in \mathrm{Id}'(T)$.

The image $\mathrm{Im} V\subset TM$ of the Lagrange anchor defines
a generalized distribution on $M$ in the sense of Sussmann
\cite{Su}. Recall that  a generalized distribution $\mathcal{V}$
is said to be \textit{involutive}, if it is closed under the Lie
bracket, that is $[\mathcal{V},\mathcal{V}]\subset \mathcal{V}$.
The \textit{involutive closure} of $\mathcal{V}$ is, by
definition, the minimal involutive distribution
$\overline{\mathcal{V}}$ containing $\mathcal{V}$ as a
subdistribution. Clearly, $\overline{\mathcal{V}}$ is spanned by
the iterated commutators of $\mathcal{V}$ and the equality
$\mathcal{V}=\overline{\mathcal{V}}$ amounts to involutivity of
$\mathcal{V}$.

\begin{defn}
A Lagrange anchor $V: \mathcal{E}^\ast\rightarrow TM$ is said to
be transitive if $\mathrm{Im}V = TM$  and weakly transitive if
$\overline{\mathrm{Im} V} = TM$.
\end{defn}

For example, the canonical Lagrange anchor $V=\mathrm{id}:
TM\rightarrow TM$ for Lagrangian equations of motion
$T=\mathrm{d}S=0$ is both integrable and transitive.

\begin{prop}\label{prop2}
The Lagrange anchor $V: \mathcal{E}^\ast\rightarrow TM$ takes
identities to proper symmetries. If the Lagrange anchor is weakly
transitive, then each proper symmetry comes from some identity.
\end{prop}

\begin{proof}
The first statement follows from  the identity
\begin{equation}\label{}
    L_{\Psi}T=\mathrm{d}_{\mathrm{\mathcal{E}}}(i_\Psi
    T)+i_{\Psi}(\mathrm{d}_{\mathcal{E}}T)=0 \qquad \forall \Psi\in
    \mathrm{Id}'(T)\,.
\end{equation}
This identity also implies that
\begin{equation}\label{diT}
    \mathrm{d}_{\mathcal{E}}(i_\Psi T)=0
\end{equation}
for any generator of proper symmetry $\Psi$. If the function
$i_\Psi T$ is annihilated by the anchor distribution $\mathrm{Im}
V$, then it has to be  annihilated by the involutive closure
$\overline{\mathrm{Im}{V}}$. Since for  transitive Lagrange
anchors $\overline{\mathrm{Im}{V}}=TM$, we have
$\mathrm{d}(i_{\Psi}T)=0$. Thus, $\langle\Psi, T\rangle\doteq 0$
and $\Psi\in \mathrm{Id}'(T)$.
\end{proof}

Let $\mathrm{PSym}'_V(T)\subset \mathrm{Sym}(T)$ denote the
subspace of all proper symmetries and let $\mathrm{PSym}_V(T)$ be
the quotient space of all proper symmetries by the trivial and
gauge symmetries. According to this definition
$\mathrm{PSym}_V(T)$ is a subspace in $\mathrm{RSym}(T)$.
Combining Proposition \ref{prop2} with Proposition \ref{prop1}, we
arrive at the following statement.

\begin{prop}\label{WT}
The bundle map $V: \mathcal{E}^\ast\rightarrow TM $ induces a well
defined homomorphism
\begin{equation}\label{Vstar}
    H(V): \mathrm{Char}(T)\rightarrow \mathrm{PSym}_V(T)\,.
\end{equation}
If the equivalence class $H(V)$ contains a weakly transitive
Lagrange anchor, then the homomorphism (\ref{Vstar}) is
surjective.
\end{prop}

The symmetries of a classical system form the Lie algebra with
respect to the vector-field commutator. On the other hand, given a
Lagrange anchor, the dual of the dynamics bundle carries the
bracket operation (\ref{LB}). Having in mind the Lagrangian case,
where  characteristics coincide with the symmetries of action, one
may expect a certain relationship between the aforementioned
multiplicative structures in $\Gamma(TM)$ and
$\Gamma(\mathcal{E}^\ast)$.  This is detailed in the following
proposition.

\begin{prop}
If $V(\Psi)$ and $V(\Psi')$ are two proper symmetries, then
$V([\Psi,\Psi'])$ is also a proper symmetry. Furthermore, the
subspace of identities $\mathrm{Id}'(T)\subset
\Gamma(\mathcal{E}^\ast)$ is closed with respect to the  bracket
(\ref{LB}), the trivial and Noether identities forming an ideal in
$\mathrm{Id}'(T)$. If the Lagrange structure is integrable, then
$V:\Gamma(\mathcal{E}^\ast)\rightarrow \Gamma(TM)$ is a Lie
algebra homomorphism and the bracket (\ref{LB}) makes
$\mathrm{Id}'(T)$ and $\mathrm{Char}(T)$ into the Lie algebras. In
the latter case, we have a well defined Lie algebra homomorphism
(\ref{Vstar}).
\end{prop}

\begin{proof}
Applying the fourth identity (\ref{Rels}) to $T$ and using
(\ref{diT}), we find that $L_{[\Psi,\Psi']}T=0 $. Hence
$[\Psi,\Psi']\in \Gamma(\mathcal{E}^\ast)$ is a proper symmetry
generator.

Applying now the third identity (\ref{Rels}) to $T$, we get
$i_{[\Psi,\Psi']}T=0$ provided that $\Psi, \Psi'\in
\mathrm{Id}'(T)$. So the bracket of two identities is an identity
again. Suppose for a moment  that $\Psi'$ is given by a sum of the
Noether and trivial identities, i.e., $\Psi'=Z(\xi)+K(T)$, for
some $K\in \Gamma(\mathcal{E}^\ast \wedge \mathcal{E}^\ast)$ and
$\xi \in \Gamma(\mathcal{G})$. We have
\begin{equation}\label{}
\begin{array}{lll}
    [\Psi,\Psi']&=&L_{\Psi} K(T)+L_\Psi
    Z(\xi)\\[3mm]
    &=&(L_{\Psi}K)(T)+ Z(V(\Psi)\xi) + (L_\Psi Z)(\xi) \,.
\end{array}
\end{equation}
By definition, the right hand side is an identity. Clearly, the
first term is a trivial identity and the second one is a Noether
identity. Then the third  term, being proportional to an arbitrary
section $\xi$, has to define a Noether identity as well. We thus
conclude that the trivial and Noether identities constitute an
ideal in $\mathrm{Id}'(T)$ with respect to the bracket
multiplication. So we have a well defined bracket on the quotient
space $\mathrm{Char}(T)$.

It follows from the fourth identity in (\ref{Rels}) that
\begin{equation}\label{LAH}
    [V(\Psi),V(\Psi')]\cdot f-V([\Psi,\Psi'])\cdot f =
    W(\Psi,\Psi')\cdot f \qquad \forall f\in C^{\infty}(M)\,,
\end{equation}
where
\begin{equation}\label{}
    W(\Psi,\Psi')\cdot f\equiv \langle
    \Psi\wedge\Psi',\mathrm{d}^2_{\mathcal{E}}f\rangle\,.
\end{equation}
By construction, the vector field $W(\Psi,\Psi')$ is a symmetry,
not necessary proper though.  It vanishes for integrable Lagrange
structures, when the bracket (\ref{LB}) enjoys the Jacobi
identity. In this case, relation  (\ref{LAH}) says that
$V:\Gamma(\mathcal{E}^\ast)\rightarrow \Gamma(TM)$ is a Lie
algebra homomorphism. This homomorphism takes the subalgebra
$\mathrm{Id}'(T)$ to the subalgebra $\mathrm{PSym}'_V(T)$. Passing
to the quotient, we get a well defined Lie algebra homomorphism
(\ref{Vstar}).
\end{proof}

\vspace{4mm}
 \noindent {\textbf{Example}} (continuation). Besides the
 homomorphism $V: T^\ast M\rightarrow TM$, the Lagrange anchor (\ref{Vw})
 defines a bracket (\ref{LB}) on the cotangent bundle $T^\ast M$. The
 latter is given by
\begin{equation}\label{brackODE}
    [\delta x^i(t),\delta
    x^j(t')]=\delta x^k(t)\partial_k\alpha^{ij}(t,x(t))\delta(t-t')\,.
\end{equation}
Formulas (\ref{Vw}), (\ref{alpha}) and (\ref{brackODE}) taken
together define the most general Lagrange structure associated to
the first-order differential equations (\ref{ODE}).

Clearly, the Lagrange anchor (\ref{Vw}) is transitive if for any
$t\in \mathbb{R}$ the bivector $\alpha_t\in \Gamma(\wedge^2 TY)$
is nondegenerate. In this case, each fiber $(Y, \alpha_t)$ of
$Y\times \mathbb{R}$ is an {almost symplectic manifold}. One can
also check \cite{KLS} that the integrability of the Lagrange
anchor (\ref{Vw}) amounts to the Jacobi identity
\begin{equation}\label{}
    \frac16 [\alpha,\alpha]\equiv \alpha^{im}\partial_m
    \alpha^{jk}\partial_i\wedge\partial_j\wedge\partial_k=0\,.
\end{equation}
Here the square brackets denote the Schouten commutator of
polyvector fields. In the integrable case, the bundle map
$\alpha_t: T^\ast Y\rightarrow TY$ defines the Lie algebroid of
the Poisson manifold $(Y,\alpha_t)$. If the Lagrange anchor
(\ref{Vw}) is both transitive and integrable, then $(Y, \alpha_t)$
is a symplectic manifold.

Suppose now that for any $t\in \mathbb{R}$ the involutive closure
of the vector distribution $\mathrm{Im}\alpha_t\subset TY$
coincides with the whole tangent bundle $TM$. Then the Lagrange
anchor (\ref{Vw}), (\ref{alpha}) is weakly transitive. The
condition for the 1-form (\ref{psi}) to be a proper symmetry of
(\ref{ODE}) gives
\begin{equation}\label{trans}
    i_X(\widetilde{d}\psi)=0\,,\qquad
i_X(\partial_t\psi+{\widetilde{d}}\psi(v))=0
\end{equation}
for any vertical vector field $X=\alpha(w)$, with $w$ being a
1-form. It follows from the first equation in (\ref{trans}) that
\begin{equation}\label{}
    L_X\widetilde{d}\psi=0\quad \Rightarrow\quad \widetilde{d}\psi=0\,,
\end{equation}
since the vector fields $X=\alpha_t(w)$ generate the whole Lie
algebra $\frak{X}(Y)$. If the first group of the De Rham
cohomology of $Y$ is trivial, then $\psi=\widetilde{d}f'$ for some
function $f'\in C^\infty(Y\times \mathbb{R})$ and the second
equation in (\ref{trans}) takes the form
\begin{equation}\label{}
    X\cdot (\partial_t f'+v\cdot f')=0\,.
\end{equation}
Again, the transitivity implies that  $\partial_t f'+v\cdot
f'=\dot g(t)$ for some $g\in C^\infty(\mathbb{R})$. Setting
$f=f'-g$, one can see that the 1-form $\psi=\widetilde{d}f$
defines a characteristic for (\ref{ODE}) in accordance with
(\ref{vf}). So, every proper symmetry for the transitive Lagrange
anchor (\ref{Vw}) comes from some characteristic. This is in
agreement  with Proposition \ref{WT}. In the case being considered
the homomorphism (\ref{Vstar}) is actually an isomorphism. Indeed,
if $X$ is the proper symmetry corresponding to a characteristic
$\Psi=\mathrm{d}\int f(t,x)dt$, then $X=0$ implies that
$\alpha(w)\cdot f=0$ for any 1-form $w$. In view of the weak
transitivity, $f$ has to be a function of $t$ alone. Then,
$\Psi=0$. We thus arrive at the following generalization of the
Noether theorem for the first-order ODEs:

\textit{Let $V$ be the weakly transitive Lagrange anchor
(\ref{Vw}) associated to the differential  equations (\ref{ODE})
and a vertical bivector field $\alpha$ on $Y\times \mathbb{R}$
with $H^1(Y)=0$, then }\begin{equation}\label{}
    \mathrm{Char}(\dot x+v)\simeq \mathrm{PSym}_{\alpha}(\dot
    x+v)\,.
\end{equation}

\vspace{4mm}
 Consider now an integrable Lagrange structure
$(\mathcal{E},T,\mathrm{d}_{\mathcal{E}})$. In the integrable
case, the dual of the dynamics bundle $\mathcal{E}^\ast\rightarrow
M$ carries the structure of Lie algebroid and the differential
$\mathrm{d}_{\mathcal{E}}: \Gamma(\wedge^\bullet
\mathcal{E})\rightarrow \Gamma(\wedge^{\bullet+1}\mathcal{E})$
makes the exterior algebra of $\mathcal{E}$-differential forms
into a cochain complex. Let $H^\bullet(\mathcal{E})=\mathrm{Ker}
\mathrm{d}_{\mathcal{E}}/\mathrm{Im} \mathrm{d}_{\mathcal{E}}$
denote the corresponding cohomology groups. Classical dynamics on
$M$ are specified by a $1$-cocycle $T$. If $T'$ is another
representative of the $\mathrm{d}_{\mathcal{E}}$-cohomology class
$[T]\in H^1(\mathcal{E})$, then
\begin{equation}\label{T'TS}
    T'=T+\mathrm{d}_{\mathcal{E}}S\,,
\end{equation}
for some $S\in C^\infty(M)$.  We refer to the last relation by
saying that the equations $T'=0$ are obtained by a \textit{proper
deformation} of the equations $T=0$; the function $S$ is called a
\textit{twist}. The relevance of this notion to the context of
conservation laws and proper symmetries is evident from the
following simple observation. If $S$ is invariant with respect to
the characteristic $\Psi\in \mathrm{Char}(T)$,
\begin{equation}\label{}
    L_{\Psi}S\doteq 0\,,
\end{equation}
then $\Psi \in \mathrm{Char}(T')$. Indeed,
\begin{equation}\label{}
i_{\Psi}T'=i_{\Psi}(T+\mathrm{d}_{\mathcal{E}}S)\doteq 0\,.
\end{equation}
Notice that the conserved current of the deformed equations $T'$
can be different from that of the equations $T$, even though the
characteristic is the same.

In general the proper deformation (\ref{T'TS}) can be quite
nontrivial and by no means  reduces to equivalence transformations
of the original equations $T$, like taking linear combinations of
the equations or change of variables. To exemplify the
non-triviality of the deformation (\ref{T'TS}), it is sufficient
to consider a proper deformations of the simplest possible
differential equations
\begin{equation}\label{xdot}
  T^i\equiv  \dot{x}{}^i(t)=0 \, .
\end{equation}
As is known, these equations enjoy an integrable Lagrange
structure given by a vertical Poisson bivector
$\alpha=\alpha^{ij}(t,x)\partial_i\wedge\partial_j$ on $Y\times
\mathbb{R}$.  Taking the twist $S$  in the form
\begin{equation}\label{HamiltonTwist}
    S= -\int dt H(t, x(t)) \, ,
\end{equation}
where $H(t, x)$ is an arbitrary function, we get the Hamiltonian
equations
\begin{equation}\label{Hamilton}
    T'^i=T^i-\alpha^{ij}\frac{\delta S}{\delta x^j} =\dot{x}^i
    -\{x^i,  H \} =0\, .
\end{equation}
So, the proper deformation is quite nontrivial as it can transform
the Hamiltonian equations with a given Hamiltonian into equations
with any other Hamiltonian and the same Poisson bracket. The
characteristics for the original equations (\ref{xdot}) have the
form $\Psi=\mathrm{d }\int f(x(t))dt$, where $f(x)$ is an
arbitrary function. If the twist (\ref{HamiltonTwist}) is
invariant with respect to $\Psi$, then $\{ f , H \} = \dot g(t)$
and the value $f - g$ is an integral of motion of
(\ref{Hamilton}).

\section{Field-theoretical examples}

To further illustrate the concept of  proper symmetry we consider
here some well-known examples of non-Lagrangian field equations.
In addition, we will see what the Lagrange anchor looks like in
local field theory.

\subsection{$p$-form  fields}
In this subsection, we consider the theory of free $p$-form fields
in the strength-tensor formalism. We will show that the notion of
Lagrange anchor enables  one to relate all the space-time
symmetries with conserved currents expressible through the
energy-momentum tensor. In the Lagrangian framework, relied on the
notion of gauge potential,  this relationship is just an immediate
consequence of the Noether theorem. Our consideration, however,
does not suppose the existence of any action functional.

Given an $n$-dimensional Riemannian manifold $(X,g)$, consider a
$p$-form field $F\in\Lambda^{p}(X)$ subject to equations
\begin{equation}\label{21}
T_1(F)\equiv dF=0\,,\qquad T_2(F)\equiv d\ast F=0\,.
\end{equation}
Here $*:\Lambda^{p}(X)\rightarrow \Lambda^{n-p}(X)$ is the Hodge
operator associated to the metric $g$.  The equations enjoy no
gauge symmetry obeying the Noether identities
\begin{equation}\label{}
    dT_1=0\,,\qquad d T_2=0\,.
\end{equation}
The last fact points up  the  non-Lagrangian nature of the field
equations (\ref{21}).\footnote{ Of course, one can solve the first
equation (\ref{21}) in terms of a gauge potential, $F=dA$,
following which the second equation (\ref{21}) becomes Lagrangian,
but we are interested in treating these equations as they are,
i.e., without passing to any equivalent Lagrangian formulation.}

To make contact with the general definitions of the previous
sections let us mention that the configuration space of fields $M$
is given here by the space $\Lambda^p(X)$ and the sections of the
dynamics bundle, in particular the equations of motion
$T=(T_1,T_2)$, take values in the space $\Lambda^{p+1}(X)\oplus
\Lambda^{n-p+1}(X)$. Since $M$ is a linear space, we can identify
the tangent space $T_FM$ at a point $F\in M$ with the space
$M=\Lambda^p(X)$ itself. Furthermore, by making use of the
standard inner product on $\Lambda(X)$,
\begin{equation}\label{IP}
    (A,B)=\int_X A\wedge * B\,,
\end{equation}
we identify the vector spaces $\Lambda^p(X)$ and
$\Lambda^{p+1}(X)\oplus \Lambda^{n-p+1}(X)$ with their dual
spaces.

One can easily check that the field equations (\ref{21}) admit the
two-parameter family of Lagrange anchors $V$ defined by
\begin{equation}\label{V}
    \langle V(W),P\rangle\dot = \langle W,V^\ast(P)\rangle =a (W_1, dP)+b(W_2,d* P)\,,
\end{equation}
for all $W=(W_1,W_2)\in \Lambda^{p+1}(X)\oplus
\Lambda^{n-p+1}(X)$, $P\in \Lambda^p(X)$, and $a,b\in \mathbb{R}$.
This family is a straightforward generalization of the Lagrange
anchor  proposed in \cite{KLS} for Maxwell's electrodynamics in
the strength-tensor formalism.

Since the equations of motion (\ref{21}) are linear and the anchor
(\ref{V}) does not depend on fields, the bracket (\ref{LB}) is
zero and the corresponding Lagrange structure appears to be
integrable. To investigate its transitivity we should consider the
kernel of the operator $V^\ast$. By definition, a $p$-form $P$
belongs to $\ker V^\ast$ if
\begin{equation}\label{VWP}
    \langle W,V^\ast(P)\rangle=0\qquad \forall
    W\,.
\end{equation}
Suppose that  $ab\neq 0$, then the last condition is equivalent to
the equations
\begin{equation}\label{dP}
    dP=0\,,\qquad d\ast P=0\,,
\end{equation}
which coincide in form with the equations of motion (\ref{21}) for
the field $F$. Since equations (\ref{dP}) enjoy no gauge symmetry,
their general solution $P$ cannot depend on arbitrary functional
parameters \footnote{For instance, if $X$ is a compact manifold
without boundary, then the solutions to Eqs. (\ref{dP}) are the
harmonic $p$-forms on $X$. These form a finite dimensional vector
space isomorphic to $H^p(X)$.}. In particular, it cannot depend on
the field $F$. We refer to this situation  by saying that the
Lagrange anchor (\ref{V}) is \textit{almost transitive} for
$ab\neq 0$.

Notice also that the Lagrange anchor (\ref{V}) is nontrivial
whenever  $a\neq b$. In case $a=b$, we can write (\ref{V}) as
\begin{equation}\label{VWPtriv}
    (W_1, T_1(aP))+ (W_2,T_2(aP))\,.
\end{equation}
The last expression  corresponds to the trivial Lagrange anchor
(\ref{Vtriv}) with $G(P)=aP$.

Recall that a vector field  $\xi\in \frak{X}(X)$ is called a
\textit{Killing vector} of the metric $g$ if $L_\xi g=0$. Each
Killing vector gives rise to both the symmetry transformation
\begin{equation}\label{}
    \delta_{\xi}F=L_\xi F
\end{equation}
of the field equations (\ref{21}) and the characteristic
$\Psi=(\Psi_1,\Psi_2)\in \Lambda^{p+1}(X)\oplus
\Lambda^{n-p+1}(X)$. The latter is given by
\begin{equation}\label{}
    \Psi_1=(-1)^{(n-p)(p-1)}\ast i_\xi \ast F\,,\qquad \Psi_2= (-1)^{p-1}\ast i_\xi F\,.
\end{equation}
It is easy to verify that
\begin{equation}\label{}
    \langle\Psi, T\rangle=(\Psi_1,T_1)+(\Psi_2,T_2)=\int_X dj\,,
\end{equation}
where the on-shell closed $(n-1)$-form $j$ reads
\begin{equation}\label{CCATF}
    j=\frac{1}{2}\left((i_{\xi} F)\wedge\ast F+(-1)^{p-1}F\wedge(i_{\xi}\ast
    F)\right)\,.
\end{equation}
In terms of local coordinates $\{x^\mu\}$ on $X$ the corresponding
conserved current is given by the 1-form
\begin{equation}\label{current}
    \ast j=\xi^{\mu}T_{\mu\nu}(x)dx^{\nu}\,,
\end{equation}
$T_{\mu\nu}=T_{\nu\mu}$ being the symmetric energy-momentum tensor
of the $p$-form field $F$.

In view of   Proposition \ref{prop2}, the characteristic $\Psi$ is
bound to be the generator of a proper symmetry.  A simple
calculation gives
\begin{equation}\label{fieldtranform}
\delta_\xi F=V(\Psi) \approx (a-b)L_{\xi}F\,.
\end{equation}
So, up to an overall constant, which can be included in $\xi$, the
transformation (\ref{fieldtranform}) coincides with the Lie
derivative of the $p$-form $F$ along the Killing vector $\xi$. As
would be expected, the transformation becomes trivial when $a=b$
(the case of trivial anchor (\ref{VWPtriv})).

To summarize,  each infinitesimal isometry of the Riemannian
manifold $(X,g)$ gives rise to a conservation law with the current
(\ref{current}). The Lagrange anchor (\ref{V}) allows us to
consider this conservation law as coming from the proper symmetry
(\ref{fieldtranform}) of the field equations (\ref{21}). Moreover,
as the Lagrange anchor is almost transitive for $ab\neq 0$, one
can be sure that each proper symmetry $V(\Psi)$ results in a
conservation law provided that the generator $\Psi$ depends on
$F$.

All the above formulas hold true for conformal Killing vectors
$\xi$ in the critical dimension $n=2p$. This  follows from the
conformal invariance of the Hodge operator $\ast:
\Lambda^p(X)\rightarrow \Lambda^p(X)$.

\subsection{Self-dual $p$-form fields}
Let $(X,g)$ be a $(4k+2)$-dimensional pseudo-Riemannian  manifold
of Lorentz signature. Then the Hodge operator
$\ast:\Lambda^{2k+1}(X)\rightarrow \Lambda^{2k+1}(X)$ squares to
$+1$ on the middle forms and we have the direct sum  decomposition
$\Lambda^{2k+1}(X)=\Lambda^{2k+1}_+(X)\oplus \Lambda^{2k+1}_-(X)$,
where $\Lambda_{\pm}^{2k+1}(X)$ are the subspaces of
(anti-)self-dual $(2k+1)$-forms on $X$. The subspaces
$\Lambda_{\pm}^{2k+1}(X)$ are known to be isotropic with respect
to the standard inner product (\ref{IP}). Therefore, the inner
product defines a non-degenerate pairing between the spaces
$\Lambda^{2k+1}_+(X)$ and $\Lambda^{2k+1}_-(X)$, so that we can
regard these spaces as being dual to each other.

Consider now a self-dual $(2k+1)$-form field $H^+$ subject to the
closedness condition
\begin{equation}\label{TH}
    T(H^+)\equiv dH^{+}=0\,.
\end{equation}
Regarding this condition as equations of motion, we see that the
sections of the corresponding dynamics bundle
$\mathcal{E}\rightarrow M$ take values in  $\Lambda^{2k+2}(X)$ and
$M$ is given by $\Lambda_+^{2k+1}(X)$. The tangent space to each
field configuration $H^+\in M$ can be identified with  $M$ as
before. Equations (\ref{TH}) are known to be non-Lagrangian unless
one introduces auxiliary fields or breaks the manifest covariance
\cite{MS, FJ, HT1, CWY, Sriv, PST}. There is the  Noether identity
$dT=0$ but no gauge symmetry. It was shown in \cite{LS3} that the
theory (\ref{TH}) admits a natural Lagrange structure defined by
the anchor
\begin{equation}\label{V+}
    \langle V(W),P\rangle\dot = \langle W,V^\ast(P)\rangle=(W,dP)\,,
\end{equation}
where $W\in \Lambda^{2k+2}(X)$ and $P\in \Lambda^{2k+1}_-(X)$.
Again, the Lagrange anchor is almost transitive, since the
condition $dP=0$ implies that the anti-self-dual form $P$ does not
depend on the field $H^+$.

If  $\xi\in \frak{X}(X)$ is a conformal Killing vector of the
metric $g$, then  $\Psi=-\ast i_\xi H^+\in \Lambda ^{2k+2}(X)$ is
a characteristic for (\ref{TH}). Indeed, using the identity
$H^+\wedge L_\xi H^{+}=0$, one can readily find
\begin{equation}\label{}
\langle\Psi,T\rangle=-(\ast i_\xi H^+, dH^+)=\int_X dj\,,\qquad
j=\frac 12i_\xi H^+\wedge H^+\,.
\end{equation}
By analogy with the previous case we can define the
energy-momentum tensor $T_{\mu\nu}$ of the self-dual field $H^+$
through  the conserved current $
    \ast j=\xi^\mu T_{\mu\nu}(x)dx^\nu$. It is easy  to see that the  tensor
    $T_{\mu\nu}$ is symmetric and traceless.

According to Proposition \ref{prop2}, the anchor (\ref{V+}) takes
the characteristic $\Psi$ to a proper symmetry of the equations of
motion (\ref{TH}). As one would expect, the corresponding
variation of $H^+$ is given by the Lie derivative along the
conformal Killing vector modulo on-shell vanishing terms:
\begin{equation}\label{XIH}
    \delta_\xi H^+=V(\Psi)\approx L_\xi H^+\,.
\end{equation}

Conversely, one could start with the proper symmetry
transformation (\ref{XIH}), which presence is evident from the
very definition of the theory (\ref{TH}),  and then argue that the
generator $\Psi$ is nothing else but a  characteristic to the
field equations since the Lagrange anchor (\ref{V+}) is almost
transitive. This gives the desired relation between the space-time
symmetries of the model and the conservation laws.

\subsection{Chiral bosons in two dimensions}

Consider now a multiplet  of $N$ self-dual 1-forms $H^+_a$,
$a=1,...,N$, in two dimensional Minkowski space
$\mathbb{R}^{1,1}$. As above, the field equations  have the form
of closedness condition
\begin{equation}\label{TH2}
    T_{a}(H^+)\equiv dH_{a}^{+}=0\,.
\end{equation}
Besides the space-time symmetries, considered in the previous
subsection, equation (\ref{TH2}) have internal rigid symmetries
corresponding to the general linear transformations in the target
space of fields:
\begin{equation}\label{G}
    H_a^+\;\; \rightarrow\;\; G_a{}^b H^+_b\,,\qquad (G_a{}^b)\in
    \mathrm{GL}(N,\mathbb{R})\,.
\end{equation}

The case of two dimensions is rather special at least for three
reasons. First, the subspaces of self-dual and anti-self-dual
1-forms admit a nice geometric visualization as a pair of
transversal isotropic distributions defining the light cone at
each point of $\mathbb{R}^{1,1}$. Second, the dynamical fields
$H^+_a$ are conserved current themselves  for if we set
$j_a=H^+_a$, then
\begin{equation}\label{djT}
    dj_a=T_a\approx 0\,.
\end{equation}
Third, and most important, the Lagrange anchor (\ref{V+}) admits a
non-abelian generalization\footnote{It should be stressed that
this generalization has nothing to do with a non-abelian
deformation of the original (abelian) equations of motion
(\ref{TH2}); the field equations remain the same.}, which can be
used to connect some of the rigid symmetries (\ref{G}) with the
conservation laws (\ref{djT}). To define this generalization we
interpret  the $N$-dimensional internal space of fields $H^+$ as
the linear space underlying  a semi-simple Lie algebra
$\mathcal{G}$ with a bases $\{t^a\}$ and the commutation relations
\begin{equation}\label{}
    [t^a,t^b]=f^{ab}_c t^c\,.
\end{equation}
(Of course, such an interpretation imposes certain restrictions on
the possible values of $N$. For example, it excludes $N=2$.) Then
we can combine the fields $H^+_a$ into a single
$\mathcal{G}$-valued 1-form $H^+=H_a^+ t^a$ subject to the
equation of motion $dH^+=0$. Geometrically, one can think of $H^+$
as a section of the trivial $SO(1,1)\times G$-vector bundle over
$\mathbb{R}^{1,1}$, where  $G$ is a Lie group  with the Lie
algebra $\mathcal{G}$.

Consider now the following non-abelian generalization of the
Lagrange anchor (\ref{V+}):
\begin{equation}\label{V+2}
\langle V(W),P\rangle\dot =\langle W,V^\ast
(P)\rangle=(W^{a},dP_{a}+g[P, H^+]_a)\,.
\end{equation}
Here $P=P_at^a$ is a $\mathcal{G}$-valued anti-self-dual 1-form
and $g$ is an arbitrary constant; all the Lie algebra indices are
raised and lowered with the help of the Killing metric on
$\mathcal{G}$. In case $g=0$, we have $N$ copies of the Lagrange
anchors (\ref{V+}).  For $g\neq 0$,  the Lagrange structure
remains integrable, but the corresponding Lie bracket (\ref{LB})
on the dual of the dynamics bundle becomes nontrivial:
\begin{equation}\label{}
    [e_a(x),e_b(x')]=-gf_{ab}^c e_c(x)\delta^2(x-x')\,,\qquad
    e_a\in \Lambda^2(\mathbb{R}^{1,1})\,.
\end{equation}

The  conserved currents (\ref{djT}) correspond to the
$\mathcal{G}$-valued characteristic
$\Psi=-\ast\varepsilon_{a}t^a$:
\begin{equation}\label{CH2}
    \langle\Psi ,T
    \rangle=-(\ast\varepsilon^{a},dH^{+}_a)=\int_{X}dj\,,
    \qquad j=\varepsilon^{a}H^{+}_a\,.
\end{equation}
By Proposition \ref{prop2} this characteristic generates  a proper
symmetry transformation
\begin{equation}\label{PS2}
   \delta_{\varepsilon}H^{+}=V(\Psi)=-g[\varepsilon,H^{+}]\,,\qquad
   \varepsilon=\varepsilon_{a}t^{a}\,.
\end{equation}
Thus, we see that  the proper symmetries constitute a subgroup
$\mathrm{Ad}(G)\subset \mathrm{GL}(N,\mathbb{R})$ in the group of
all internal symmetries (\ref{G}) of the equations of motion.

As to the relativistic symmetries of the system (\ref{TH2}), these
are proper and have the same form as in the abelian case
(\ref{XIH}).

\section{Conclusion}
Let us make some concluding remarks on  the paper results and
notice the open questions left for the future studies.

In the work \cite{KLS}, the concept of Lagrange anchor was
introduced to formulate the path-integral quantization for not
necessarily Lagrangian field theories. In the present paper, we
show that another important property of the Lagrangian dynamics -
the relationship between rigid symmetries and conservation laws -
extends to non-Lagrangian field equations whenever they are
endowed with a Lagrange anchor. It was also shown in \cite{KLS}
that every Lagrange anchor gives rise to and can be related with a
certain BRST complex on the ghost-extended configuration space of
fields. In the Lagrangian gauge theory, this complex boils down to
the well-known  BRST complex associated with the
Batalin-Vilkovisky master action. The unique existence of the
local BV master action was proven long ago \cite{H2} under the
assumptions that the original Lagrangian was local and the gauge
symmetry was finitely reducible. Unlike the Lagrangian theory, the
locality of both the Lagrange anchor and the non-Lagrangian
equations does not necessarily result in the locality of the
corresponding quantum BRST differential. Given a local Lagrange
anchor for local field equations, the obstructions to the
existence of a local  BRST complex will be studied in our next
paper. Remarkably, these obstructions, being crucial for the
quantum theory, do not matter for the study of classical dynamics.
At the classical level, the existence of a local Lagrange anchor
is sufficient to link the conservation laws with rigid symmetries,
as it is seen from the present paper.

In the next work we also plan to develop the cohomological
description of the rigid symmetries and conservation laws, making
use of the BRST formalism for non-Lagrangian gauge theories
\cite{KLS}. This will extend the cohomological formulation of the
Noether theorem, well studied in the Lagrangian setting
\cite{BBH}, to not necessarily Lagrangian dynamics. In the entire
BRST cohomology, we will identify the characteristic cohomology as
well as  the cohomology group capturing the rigid symmetries of
non-Lagrangian equations of motion. Unlike the Lagrangian theory,
these two groups of the BRST cohomology are not to be generally
isomorphic to each other.

One can see that many important ingredients of the Lagrangian
dynamics split into different categories in the non-Lagrangian
setting. The local Lagrange anchor  connects these categories,
although not always making them identical. As a result, the
Lagrange anchor endows the non-Lagrangian local field theory with
some important features of the Lagrangian one. In particular, it
offers a path-integral quantization of classical dynamics and
establishes a (partial) connection between the rigid symmetries
and the conservation laws.

\end{document}